\documentclass{llncs}

\usepackage{amsmath}
\usepackage{amsfonts}
\usepackage{amssymb}
\usepackage{stmaryrd}
\usepackage{graphicx}
\usepackage{algorithm,algorithmic}

\usepackage{color}

\newcommand{\rrangle}{\rangle\!\rangle}
\newcommand{\llangle}{\langle\!\langle}
\newcommand{\strat}[1]{\llangle#1\rrangle}                         

\newcommand{\word}[1]{\langle #1\rangle}                           
\newcommand{\set}[1]{\{#1\}}
\newcommand{\LT}[1]{\stackrel{#1}{\rightarrow}}

\newcommand{\Nat}{\mathbb{N}}
\newcommand{\A}{\textit{Act}}

\newcommand\Prop{\textit{Prop}}
\newcommand\Sim{\sqsubseteq}                                       

\newcommand\dist{\mathcal{D}}

\newcommand\step{\delta}                                           

\renewcommand\L{\mathcal{L}}

\newcommand\G{\mathcal{G}}                                         

\newcommand\lift[1]{\overline{#1}}                               
\newcommand\R{\mathcal{R}}                                       
\newcommand\Supp[1]{\lceil#1\rceil}                              
\newcommand\pd[1]{\overline{#1}}                                 

\newcommand\pone{\texttt{I}}
\newcommand\ptwo{\texttt{I\!I}}

\newcommand{\porder}{\preceq}
\newcommand{\power}{\mathcal{P}}
\newcommand{\finer}{\lhd}
\newcommand{\Id}{\texttt{Id}}
\newcommand{\upclose}[1]{\lfloor #1 \rfloor_{\porder}}
\newcommand{\bigO}{\mathcal{O}}
\newcommand\commentout[1]{}

\newcommand{\CanFollow}{\textsf{CanFollow}}
\newcommand{\CanSim}{\textsf{CanSim}}

\newcommand{\forEach}{\textbf{for\ each}}
\newcommand{\while}{\textbf{while}}
\newcommand{\doo}{\textbf{do}}

\begin{document}

\title{An Algorithm for Probabilistic~Alternating~Simulation}

\author{
  Chenyi Zhang\inst{1} \and
  Jun Pang\inst{2}
}
\institute{
School of Information Technology and Electrical Engineering,
University of Queensland, Australia
\and
Faculty of  Science, Technology and Communication,
University of Luxembourg, Luxembourg
}

\maketitle
\begin{abstract}
In probabilistic game structures, probabilistic alternating simulation (PA-simulation)
relations preserve formulas defined in probabilistic alternating-time temporal logic
with respect to the behaviour of a subset of players.
We propose a partition based algorithm for computing the largest PA-simulation.
It is to our knowledge the first such algorithm that works in polynomial time.
Our solution extends the generalised coarsest partition problem (GCPP)
to a game-based setting with mixed strategies.
The algorithm has higher complexities than those in the literature for non-probabilistic simulation
and probabilistic simulation without mixed actions, 
but slightly improves the existing result
for computing probabilistic simulation with respect to mixed actions.
\end{abstract}

\section{Introduction}
\label{sec:intro}
Simulation and bisimulation relations are useful tools in the verification
of finite and infinite state systems.
State space minimisation modulo these relations
is a valuable technique to fight the state explosion problem in model checking,
since bisimulation preserves properties formulated in logics like CTL and CTL$^*$~\cite{CE81}
while simulation preserves the universal (or safe) fragment of these logics~\cite{GL94}.

In some situations, however, 
it is necessary to model quantitative aspects of a system.
It is the case, for instance, in wireless networks, where
we often need to assume that there is a chance of connection failure with a
given rate. This requires modelling network systems with randomised behaviours
(e.g., by pooling a connection after uncertain amount of time to minimise conflict). Another important fact of real-world systems
is that environment changes, such as unexpected power-off, are often unpredictable.
Therefore, we need to encode appropriate system behaviours to handle such situations, and in order to do so,
it is sometimes crucial to employ probabilistic strategies to achieve 
the best possible outcomes~\cite{NM47}. One simple example is the rock-scissor-paper game
where there is no deterministic strategy to win since the other player's move is unknown, but there is 
a probabilistic strategy, sometimes called \emph{mixed strategy}, to win at least a third 
of all cases in a row, regardless of what the other player does.\footnote{A
mixed strategy also ensures an eventual win but deterministic strategies do not.}

A probabilistic game structure (PGS) is a model that has probabilistic transitions, 
and allows the consideration of probabilistic choices of players. The simulation relation
in PGSs, called probabilistic alternating simulation (PA-simulation), 
has been shown to preserve a fragment of probabilistic alternating-time temporal logic (PATL)
under \emph{mixed strategies}, which is used in characterising what a group of players 
can enforce in such systems~\cite{ZP10}.
In this paper we propose a polynomial-time algorithm for computing the largest 
PA-simulation, which is, to the best of our knowledge, the first algorithm for computing a simulation
relation in probabilistic concurrent games.
A PGS combines the modelling of probabilistic transitions
from probabilistic automata (PA), and the user interactions from concurrent game structures (GS).
In PA, the probabilistic notions of simulation preserve PCTL safety formulas~\cite{SL95}. 
The \emph{alternating simulation}~\cite{AHKV98} in GS has been been proved to preserve a fragment
of ATL$^*$, under the semantics of \emph{deterministic strategies}. These
simulation relations are computable in polynomial time for finite systems~\cite{Zhang08,AHKV98}.

\paragraph{Related work.}
Efficient algorithms have been proposed
for computing the largest simulation (e.g., see~\cite{HHK95,TC01,BG03,GPP03,GP08})
in finite systems,
with a variety of time and space complexities.
In particular, Gentilini et al.~\cite{GPP03} develop an efficient algorithm with an improved time complexity based on
the work of Henzinger~et~al.~\cite{HHK95} without 
losing the optimal space complexity.
Van Glabbeek and Ploeger~\cite{GP08} later find a flaw in~\cite{GPP03} and
propose a non-trivial fix.
The best algorithm for simulation in terms of time complexity is~\cite{RT07}.
To compute probabilistic simulation, Baier et al.~\cite{BEM00}
reduce the problem of establishing a weight function for the lifted relation to a maximal flow problem~\cite{AMO93}.
Cattani and Segala~\cite{CS02} reduce the problem of deciding strong probabilistic bisimulation
to LP~\cite{Sch86} problems.
Zhang et al.~\cite{ZHEJ08} develop algorithms with improved time complexity
for probabilistic simulations, following~\cite{BEM00,CS02}.
Crafa and Ranzato~\cite{CR11} improve the time complexity of the algorithms of Zhang et al.~\cite{ZHEJ08}
by applying abstract interpretation.
A space efficient probabilistic simulation algorithm is proposed by Zhang~\cite{Zhang08}
using the techniques proposed in~\cite{GPP03,GP08}.

Studies on stochastic games have actually been carried out since as early as the $1950$s~\cite{Shapley53}, 
and a rich literature has developed in recent years (e.g.~see~\cite{AHK98,Alf03,AM04,CDH06}).
One existing approach called game metrics~\cite{dAMRS08} defines approximation-based
simulation relations, with a kernel simulation characterising the logic
quantitative $\mu$-calculus ($q\mu$)~\cite{Alf03}, an extension of modal $\mu$-calculus~\cite{Kozen83}
where each state is assigned a quantitative value in $[0,1]$ for every formula.
However, so far the best solutions in the literature on approximating the simulation as defined in
the metrics for concurrent games potentially take exponential time~\cite{CdAMR10}. 
Although PA-simulation is strictly stronger than the kernel simulation relation of the game metrics in~\cite{dAMRS08},
the algorithm presented in the paper has a more tractable complexity result,
and we believe that it will benefit the abstraction or refinement based
techniques for verifying game-based~properties. 

\paragraph{Structure of the paper.}
Sect.~\ref{sec:pgs} defines basic notions that are used in the technical part.
In Sect.~\ref{sec:gcpp} we propose a solution of calculating largest PA-simulation in finite PGSs, based on
GCPP. The algorithms on PA-simulation is presented in Sect.~\ref{sec:pa-sim-alg}.
We conclude the paper in Sect.~\ref{sec:conclusion}.

\section{Preliminaries}
\label{sec:pgs}
Probabilistic game structures are defined in terms of discrete probabilistic distributions.
A \emph{discrete probabilistic distribution} $\Delta$ over a finite set $S$
is a function of type $S\rightarrow[0,1]$, where $\sum_{s\in S}\Delta(s)=1$.
We write $\dist(S)$ for the set of all such distributions on a fixed $S$.
For a set $T\subseteq S$, define $\Delta(T)=\sum_{s\in T}\Delta(s)$.
Given a finite index set $I$, a list of distributions $(\Delta_i)_{i\in I}$ 
and a list of probabilities $(p_i)_{i\in I}$
where, for all $i\in I$, $p_i\in[0,1]$ and $\sum_{i\in I}p_i=1$,
$\sum_{i\in I}p_i\Delta_i$ is obviously also a distribution.
For $s\in S$, $\pd{s}$ is called a \emph{point (or Dirac) distribution} satisfying $\pd{s}(s)=1$ and
$\pd{s}(t)=0$ for all $t\neq s$. Given $\Delta\in\dist(S)$, we define $\Supp{\Delta}$
as the set $\set{s\in S\mid\Delta(s)>0}$, which is the \emph{support}~of~$\Delta$.

In this paper we assume a set of two players $\set{\pone, \ptwo}$ (though our results can be extended
to handle a finite set of players as in the standard game structure and ATL semantics~\cite{AHK02}),
and $\Prop$ a finite set of propositions.

\begin{definition}\label{def:PGS}
A probabilistic game structure 
$\G$ is a tuple $\word{S, s_0, \L, \A,\step}$, where
\begin{itemize}
\item $S$ is a finite set of states, with $s_0$ the initial state;
\item $\L:S\rightarrow 2^{\Prop}$ is the labelling function which
      assigns to each state $s\in S$ a set of propositions that are true in $s$;
\item $\A=\A_\pone\times\A_\ptwo$ is a finite set of joint actions, where $\A_\pone$ and $\A_\ptwo$ are, respectively, the sets of
      actions for players $\pone$ and $\ptwo$;
\item $\step:S\times\A\rightarrow \dist(S)$ is a transition function.
\end{itemize}
\end{definition}

If in state $s$ player $\pone$ performs action $a_1$ and player $\ptwo$ performs action $a_2$
then $\step(s,\word{a_1,a_2})$ is the distribution for the next states.
During each step the players choose their next moves simultaneously.
We define a \emph{mixed action} of player $\pone$ ($\ptwo$) as a distribution over $\A_\pone$ ($\A_\ptwo$),
and write $\Pi_\pone$ ($\Pi_\ptwo$) for the set of mixed actions of player $\pone$  ($\ptwo$).\footnote{
Note $\Pi_\pone$ is equivalent to $\dist(\A_\pone)$, though we choose a different symbol because
the origin of a mixed action is a simplified \emph{mixed strategy} of player $\pone$ which has type
$S^+\rightarrow\dist(\A_\pone)$. A mixed action only considers player $\pone$'s current step.
} In particular,
$\pd{a}$ is a \emph{deterministic} mixed action which always chooses $a$.
We lift the transition function $\step$ to handle mixed actions. Given $\pi_1\in\Pi_\pone$ and $\pi_2\in\Pi_\ptwo$,
for all $s, t\in S$, we have
\[\lift{\step}(s,\word{\pi_1,\pi_2})(t)=\sum_{a_1\in\A_\pone,a_2\in\A_\ptwo}\pi_1(a_1)\cdot\pi_2(a_2)\cdot\step(s,\word{a_1,a_2})(t)\]

\begin{figure}[tp]
\centering
\includegraphics[height=3.5cm]{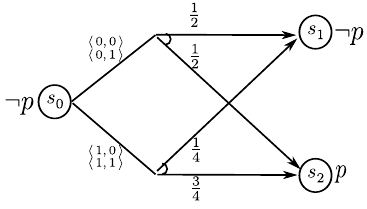}
\caption{A probabilistic game structure.}\label{fig:simplegs}
\end{figure}

\begin{example}\label{example:pgs}
Assume $\Prop=\set{p}$. A simple PGS with the initial state $s_0$ in Fig.~\ref{fig:simplegs}
can be defined as  $\G=\word{S, s_0, \L, \A,\step}$, where
\begin{itemize}
\item $S=\set{s_0,s_1,s_2}$;
\item $\L(s_0)=\L(s_1)=\emptyset$ and $\L(s_2)=\set{p}$;
\item $\A_\pone=\A_\ptwo=\set{0,1}$;
\item $\step(s_0,\langle 0,0\rangle)=\step(s_0,\langle 0,1\rangle)=\Delta$ with $\Delta(s_1)=\Delta(s_2)=\frac{1}{2}$ and
         $\step(s_0,\langle 1,0\rangle)=\step(s_0,\langle 1,1\rangle)=\Delta'$ with $\Delta'(s_1)=\frac{1}{4}$, $\Delta'(s_2)=\frac{3}{4}$;
\item $\step(s_i,a)=\pd{s_i}$ for $i\in\set{1,2}$ and $a\in\A$ ($s_1$ and $s_2$ are \emph{absorbing} states).
\end{itemize}
\end{example}

\begin{definition}\label{def:action-combine}
Given a list of mixed actions $\word{\pi_i}_{i\in I}$ (of player $\pone$),
$\word{p_i}_{i\in I}$ satisfying $\sum_{i\in I}p_i=1$,
$\sum_{i\in I}p_i\pi_i$ is a mixed action defined by $\left(\sum_{i\in I}p_i\pi_i\right)(s)(a)=\sum_{i\in I}p_i\cdot\left(\pi_i(s)(a)\right)$
for all $s\in S$ and $a\in\A_\pone$.
\end{definition} 

\begin{lemma}\label{lem:split:mix-action}
Let $s\in S$, $\pi\in\Pi_\pone$ and $\sigma=\sum_{i\in I}p_i\sigma_i\in\Pi_\ptwo$,
then $\lift\step(s,\word{\pi,\sigma})=\sum_{i\in I}p_i\cdot\lift\step(s,\word{\pi,\sigma_i})$. 
\end{lemma}
\begin{proof}
Let $t\in S$, then
\[
\begin{array}{lllr}
     && \lift\step(s, \word{\pi,\sigma})(t) \\
    &=& \sum_{a_1\in\A_1}\sum_{a_2\in\A_2}\pi(s)(a_1)\cdot\sigma(s)(a_2)\cdot\step(s,\word{a_1,a_2})(t) \\
    &=& \sum_{a_1\in\A_1}\sum_{a_2\in\A_2}\pi(s)(a_1)\cdot\sum_{i\in I}p_i\cdot\sigma_i(s)(a_2)\cdot\step(s,a_1,a_2)(t) \\
    &=& \sum_{i\in I}p_i\cdot\left(\sum_{a_1\in\A_1}\sum_{a_2\in\A_2}\pi(s)(a_1)\cdot\sigma_i(s)(a_2)\cdot\step(s,a_1,a_2)(t)\right) \\
    &=& \sum_{i\in I}p_i\cdot\lift\step(s,\word{\pi,\sigma_i})(t)
\end{array}
\]
\end{proof}

\noindent
The proof of the following lemma is similar to the above.
\begin{lemma}\label{lem:split:mix-action-2}
Let $s\in S$, $\pi=\sum_{i\in I}p_i\pi_i\in\Pi_\pone$ and $\sigma\in\Pi_\ptwo$,
then $\lift\step(s,\word{\pi,\sigma})=\sum_{i\in I}p_i\cdot\lift\step(s,\word{\pi_i,\sigma}$. 
\end{lemma}

Simulation relations in probabilistic systems
require a definition of \emph{lifting}~\cite{JL91}, which extends the relations to the domain of distributions.\footnote{
In a probabilistic system without explicit user interactions, state $s$ is simulated by state $t$ if for every $s\LT{a}\Delta_1$
there exists $t\LT{a}\Delta_2$ such that $\Delta_1$ is simulated by $\Delta_2$.
} 
Let $S$, $T$ be two sets and $\R\subseteq S\times T$ be a relation,
then $\lift\R\subseteq\dist(S)\times\dist(T)$ is a \emph{lifted relation} defined by $\Delta\,\lift\R\,\Theta$
if there exists a weight function $w:S\times T\rightarrow[0,1]$ such that
\begin{itemize}
\item $\sum_{t\in T}w(s, t)=\Delta(s)$ for all $s\in S$,
\item $\sum_{s\in S}w(s, t)=\Theta(t)$ for all $t\in T$,
\item $s\R\, t$ for all $s\in S$ and $t\in T$ with $w(s,t)>0$.
\end{itemize}

The intuition behind the lifting is that each state in the support of one distribution may correspond
to a number of states in the support of the other distribution, and vice versa.
The example in Fig.~\ref{fig:lifting} is taken from~\cite{Seg95} to show how to lift one relation.
We have two set of states $S=\set{s_1,s_2}$ and $T=\set{t_1,t_2,t_3}$,
and $\R=\set{(s_1,t_1),(s_1,t_2),(s_2,t_2),(s_2,t_3)}$.
We have $\Delta\,\lift\R\,\Theta$, where $\Delta(s_1)=\Delta(s_2)=\frac{1}{2}$
and $\Theta(t_1)=\Theta(t_2)=\Theta(t_3)=\frac{1}{3}$.
To check this, we define a weight function $w$ by:
$w(s_1,t_1)=\frac{1}{3}$, $w(s_1,t_2)=\frac{1}{6}$
$w(s_2,t_2)=\frac{1}{6}$, and $w(s_2,t_3)=\frac{1}{3}$.
The dotted lines indicate the allocation of weights required to relate
$\Delta$ to $\Theta$ via $\lift\R$.
By lifting in this way, we are able to extend the notion of alternating simulation~\cite{AHKV98}
to a probabilistic setting.

\begin{figure}[tp]
\centering
\includegraphics[height=4.0cm]{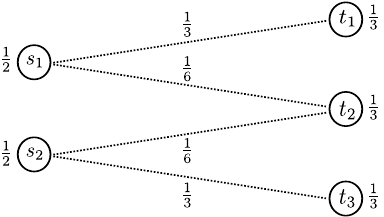}
\caption{An example showing how to lift one relation.}\label{fig:lifting}
\end{figure}

We present a property of lifted relations.
that by combining distributions that are lift-related with the same weight on both sides
we get the resulting distributions lift-related.

\begin{lemma}\label{lem:lift-comb}
Let $\R\subseteq S\times S'$ and $\word{p_i}_{i\in I}$ be a list of values satisfying
$\sum_{i\in I}p_i=1$, $\Delta_i\,\lift\R\,\Delta'_i$ for $\Delta_i\in\dist(S)$ and $\Delta_i'\in\dist(S')$
for all $i$, then $\sum_{i\in I}p_i\Delta_i\,\lift\R\,\sum_{i\in I}p_i\Delta'_i$.
\end{lemma}

Now we present the definition of PA-\pone-simulation.
\begin{definition}
\label{def:sim}
Given a PGS, 
a \emph{probabilistic alternating $\pone$-simulation} (PA-\pone-simulation)
is a relation $\Sim\ \subseteq S\times S$ such that
if~$s\Sim t$,~then
\begin{itemize}
\item $\L(s)=\L(t)$,
\item for all $\pi_1\in\Pi_\pone$, there exists
      $\pi_1'\in\Pi_\pone$, such that for all $\pi_2'\in\Pi_\ptwo$, there exists $\pi_2\in\Pi_\ptwo$, such that
      $\pd\step(s, \word{\pi_1, \pi_2})\mathrel{\lift{\Sim}}\pd\step(t, \word{\pi_1',\pi_2'})$.
\end{itemize}
\end{definition}
If $s$ PA-\pone-simulates $t$ and $t$ PA-\pone-simulates $s$, we say $s$ and $t$ are
\emph{PA-\pone-simulation equivalent}.\footnote{Alternating simulations and equivalences are for player \pone\
unless stated otherwise.}

PA-\pone-simulation has been shown to preserve a fragment of
PATL which covers the ability of player \pone\ to enforce certain temporal
requirements~\cite{ZP10}. For example, if in state $s$ player \pone\ can enforce reaching some states satisfying $p$
within $5$ transition steps and with probability at least $\frac{1}{2}$, written
$s\models\strat{\pone}^{\geq\frac{1}{2}}\Diamond^{\leq 5}p$, then for every state $t$ that
simulates $s$ with respect to \pone, i.e., $s\Sim t$ by some PA-\pone-simulation `$\Sim$',
we also have $t\models\strat{\pone}^{\geq\frac{1}{2}}\Diamond^{\leq 5}p$.

\subsection*{General Coarsest Partition Problem}
The general coarsest partition problem (GCPP) provides a characterisation of (non-probabilistic) simulation
in finite state transition systems~\cite{GPP03}. Informally, in this approach, states that are (non-probabilistic) simulation equivalent
are grouped into the same block, and all such blocks form a partition over the (finite) state space.
Based on the partition, blocks are further related by a partial order $\porder$, so that if $P\porder Q$, 
then every state in block $P$ is simulated by every state in block $Q$. The GCPP is to find, for a given PGS, 
the smallest such set of blocks.
In the literature such a methodology yields space efficient algorithms for computing the largest (non-probabilistic)
simulation relation in a finite system~\cite{GPP03,GP08}.
Similar methods have been adopted and developed to compute the largest simulation relations in the model
of probabilistic automata~\cite{Zhang08}.

We briefly review the basic notions that are required to present the GCPP problem.
A \emph{partition} over a set $S$, is a collection $\Sigma\subseteq\power(S)$ satisfying
(1) $\bigcup\Sigma=S$ and (2) $P\cap Q=\emptyset$ for all distinct \emph{blocks} $P,Q\in\Sigma$.
Given $s\in S$, write $[s]_\Sigma$ for the block in partition $\Sigma$ that contains $s$.
A partition $\Sigma_1$ is \emph{finer} than $\Sigma_2$, written $\Sigma_1\finer\Sigma_2$, if for all $P\in\Sigma_1$
there exists $Q\in\Sigma_2$ such that $P\subseteq Q$.

Given a set $S$, a \emph{partition pair} over $S$ is
$(\Sigma,\porder)$ where $\Sigma$ is a partition over $S$ and $\porder\ \subseteq\Sigma\times\Sigma$
is a partial order. Write $Part(S)$ for the set of partition pairs on $S$.
If $\Upsilon\finer\Sigma$ and $\porder$ is a relation on $\Sigma$,
then $\porder(\Upsilon)=\set{(P,Q)\mid P,Q\in\Upsilon,\exists P',Q'\in\Sigma,P\subseteq P',Q\subseteq Q',P'\porder Q'}$
is the relation on $\Upsilon$ \emph{induced} by $\porder$.
Let $(\Sigma_1,\porder_1)$ and $(\Sigma_2,\porder_2)$ be partition orders,
write $(\Sigma_1,\porder_1)\leq(\Sigma_2,\porder_2)$ if $\Sigma_1\finer\Sigma_2$,
and $\porder_1\subseteq\porder_2(\Sigma_1)$.
Define a relation $\Sim_{(\Sigma,\porder)}\subseteq S\times S$ as determined by a partition pair $(\Sigma,\porder)$
by $s\Sim_{(\Sigma,\porder)}t$ iff $[s]_\Sigma\porder[t]_\Sigma$.

Let $\rightarrow\subseteq S\times S$ be a (transition) relation and $\L:S\rightarrow 2^\Prop$ a labelling function,
then a relation $\sqsubseteq$ is a simulation on $S$ if for all $s,t\in S$ with $s\sqsubseteq t$,
we have (1) $\L(s)=\L(t)$ and (2) $s\rightarrow s'$ implies $t\rightarrow t'$ and $s'\sqsubseteq t'$.
Let $(\Sigma,\porder)$ be a partition pair on $S$, then it is \emph{stable} with respect to 
$\rightarrow$ if for all $P,Q\in\Sigma$ with $P\porder Q$ and $s\in P$ such that $s\rightarrow s'$ with $s'\in P'\in\Sigma$,
then there exists $Q'\in\Sigma$ such that for all $t\in Q$, there exists $t'\in Q'$ such that $t\rightarrow t'$.
The following result is essential to the GCPP approach, as we derive the largest simulation relation
by computing the coarsest partition pair over a finite state space.\footnote{
We choose the word \emph{coarsest} for partition pairs to make it consistent with the standard term GCPP, and it is clear 
in the context that \emph{coarsest} carries the same meaning as \emph{largest} with respect to the order $\leq$
defined on partition~pairs.}
\begin{proposition}\label{prop:gcpp}~\cite{GPP03,GP08}
Let $(\Sigma,\porder)$ be a partition pair, then it is stable with respect to $\rightarrow$ iff the induced relation
$\Sim_{(\Sigma,\porder)}$ is a simulation (with respect to $\rightarrow$).
\end{proposition}

Given a transition relation on a state space there exists a unique largest simulation relation. Thus, solutions
to GCPP provide the coarsest stable partition pairs, and they have been proved to characterise the largest simulation
relations in non-probabilistic systems~\cite{GPP03,GP08}.

\section{Solving GCPP in Probabilistic Game Structures}
\label{sec:gcpp}
In this section we extend the GCPP framework to characterise PA-simulations in PGSs.
Given a PGS $\G=\word{S, s_0, \L, \A,\step}$, a \emph{partition pair} over $\G$ is
$(\Sigma,\porder)$ where $\Sigma$ is a partition over $S$. Write $Part(\G)$ for the
set of all partition pairs over~$S$. We show how to compute the coarsest partition pair and prove
that it characterises the largest PA-simulation for a given player.

Since in probabilistic systems transitions go from states to distributions over states, we first present
a probabilistic version of \emph{stability}, as per~\cite{Zhang08}.
Let $\rightarrow\subseteq S\times\dist(S)$ be a probabilistic (transition) relation.
For a distribution $\Delta\in\dist(S)$ and $\Sigma$ a partition, write $\Delta_\Sigma$
as a distribution on $\Sigma$ defined by $\Delta_\Sigma(P)=\Delta(P)$ for all $P\in\Sigma$.
Let $(\Sigma,\porder)$ be a partition pair, it is \emph{stable} with respect to the relation $\rightarrow$,
if for all $P,Q\in\Sigma$ with $P\porder Q$ and $s\in P$ such that $s\rightarrow\Delta$, then for all $t\in Q$
there exists $t\rightarrow\Theta$ such that $\Delta_\Sigma\,\lift{\porder}\,\Theta_\Sigma$.

Another obstacle in characterising PA-simulation is that the concerned player can only partially
determine a transition.
That is, after player \pone\ performs an action on a state, the exact future distribution on next states
depends on an action from player \ptwo. Therefore, we need to (again) lift the stability condition for
PA-\pone-simulation from distributions to sets of distributions.

Let $\leq\ \subseteq S\times S$ be a partial order on a set $S$, define
$\leq_{Sm}\subseteq\power(S)\times\power(S)$, by $P\leq_{Sm}Q$ if for all $t\in Q$ there exists $s\in P$ such that $s\leq t$.
In the literature this definition is known as a `Smyth order'~\cite{Smy78}.
In a PGS, we `curry' the transition function by defining
$\lift\step(s,\pi_1)=\set{\lift\step(s, \word{\pi_1, \pi_2})\mid\pi_2\in\Pi_\ptwo}$,
which is the set of distributions that are possible if player \pone\ takes a mixed action $\pi_1\in\Pi_\pone$ on $s\in S$.
\begin{definition} (lifted stability) \label{def:lift-stability}
Let $(\Sigma,\porder)$ be a partition pair on $S$ in a PGS, it is stable with respect to player $\pone$'s choice,
if for all $\pi\in\Pi_\pone$, $P,Q\in\Sigma$ with $P\porder Q$ and $s\in P$,
there exists $\pi'\in\Pi_\pone$ such that $\lift\step(s,\pi)_\Sigma\,\lift{\porder}_{Sm}\,\lift{\step}(t,\pi')_\Sigma$ for~all~$t\in Q$.
\end{definition}
Intuitively, the Smyth order captures the way of \emph{behavioral} simulation. That is,
if $\lift\step(t,\pi')$ is at least as restrictive as $\lift\step(s,\pi)$, then whatever player \pone\
is able to enforce by performing $\pi$ in $s$, he can also enforce it by performing $\pi'$ in $t$, as player \ptwo\ 
has \emph{fewer} choices 
in $\lift\step(t,\pi')$ than in $\lift\step(s,\pi)$. At this point, for the sake of readability, if it is
clear from the context, we write $W$ for $W_\Sigma$ as the distribution $W$ mapped onto partition~$\Sigma$.

For simulation relations, it is also required that the related states
agree on their labelling. Define $\Sigma_0$ as the \emph{labelling partition} satisfying
for all $s,t\in S$, $\L(s)=\L(t)$ iff $[s]_{\Sigma_0}=[t]_{\Sigma_0}$.
Write $Part^0(\G)\subseteq Part(\G)$ for the set of partition pairs $(\Sigma,\porder)$ satisfying
$(\Sigma,\porder)\leq(\Sigma_0,\Id)$, where $\Id$ is the identity~relation.

\begin{lemma}\label{lem:stable-pa-sim}
For all $(\Sigma,\porder)\in Part^0(\G)$, if $(\Sigma,\porder)$ is a stable partition pair with respect to player \pone's choice
then $\Sim_{(\Sigma,\porder)}$ is a PA-\pone-simulation.
\end{lemma}
\begin{proof}
Straightforward by Definition~\ref{def:lift-stability}.
\end{proof}

Obviously every PA-\pone-simulation is contained in the relation induced by $(\Sigma_0,\Id)$,
and moreover, the above lemma asserts that every stable partition pair smaller than $(\Sigma_0,\Id)$ is a PA-\pone-simulation. 
In the following, we try to compute the
coarsest partition pair by refining
$(\Sigma_0,\Id)$ until it stabilises.
The resulting stable partition pair can be proved to characterise the largest PA-\pone-simulation on the state space $S$~as~required.

We say $t$ simulates $s$ with respect to player-\pone's choice on a partition pair $(\Sigma,\porder)$ if for all $\pi\in\Pi_\pone$,
there exists $\pi'\in\Pi_\pone$ such that $\lift\step(s,\pi)\,\lift{\porder}_{Sm}\,\lift\step(t,\pi')$.
For better readability, sometimes we also say $t$ simulates $s$ on $(\Sigma,\porder)$ if it is clear from the context,
and write $s\Sim_{(\Sigma,\porder)}^\star t$. Note it is straightforward to show that $\Sim_{(\Sigma,\porder)}^\star$
is a transitive relation, by definition of $\lift{\porder}_{Sm}$.
Let $(\Sigma_1,\porder_1)\leq(\Sigma_2,\porder_2)$, we say $(\Sigma_1,\porder_1)$ is stable
on $(\Sigma_2,\porder_2)$, if for all $P,Q\in\Sigma_1$ with $P\porder_1Q$, $s\in P$ and $t\in Q$,
$t$ simulates $s$ on $(\Sigma_2,\porder_2)$.
\begin{definition}\label{def:rho}
Define an operator $\rho:Part(\G)\rightarrow Part(\G)$,
such that $\rho((\Sigma,\porder))$
is the largest partition pair $(\Sigma', \porder')\leq(\Sigma,\porder)$ that is
stable on $(\Sigma,\porder)$.
\end{definition}
The operator $\rho$ has the following properties.
\begin{lemma}\label{lem:rho-well-defined}
$\rho$ is well defined on $Part(\G)$.
\end{lemma}
\begin{proof}
We show that given a partition pair $(\Sigma,\porder)$ on $S$, $\rho((\Sigma,\porder))$ is a unique partition pair.
Let $P\in\Sigma$.
Define $\leq_P\subseteq P\times P$ by $s\leq_P t$ if $s\Sim_{(\Sigma,\porder)}^\star t$.
Then $\leq_P$ is a preorder on $P$, from which we define a partition pair
$(\Sigma_P, \porder_P)$ where
$\Sigma_P=\set{\set{t\in P\mid s\leq_Pt\wedge t\leq_Ps}\mid s\in P}$
and $X_1\,\porder_P\,X_2$
if there exist $s\in X_1$ and $t\in X_2$ such that $s\leq_P t$.
Define $\rho((\Sigma,\porder))=(\Sigma', \porder')$ with $\Sigma'=\bigcup_{P\in\Sigma}\Sigma_P$ and
$\porder'=(\porder\setminus Id)(\Sigma')\cup\bigcup_{P\in\Sigma}\porder_P$.
For the definition of $\porder'$, the first part of the union $(\porder\setminus Id)(\Sigma')$ is the relation on $\Sigma'$ as induced from the nonreflexive part of $\porder$, and in the second part each $\porder_P$ gives a new relation generated inside  block $P$ which is stable on $(\Sigma,\porder)$. Note that each $\porder_P$ is acyclic, and thus a partial order on $\Sigma_P$.
This implies that $\porder'$ is a partial order on $\Sigma'$.

We show that $(\Sigma', \porder')$ is indeed the largest such partition pair.
Suppose there exists $(\Sigma'', \porder'')$ such that $(\Sigma'', \porder'')\leq(\Sigma,\porder)$ and
it is stable on $(\Sigma,\porder)$,
we show that $(\Sigma'', \porder'')\leq(\Sigma',\porder')$.
\begin{itemize}
\item Let $P\in\Sigma''$ and $s\in P$, then there exists $P'\in\Sigma'$ such that $s\in P'$. First we have
      $P\subseteq[s]_\Sigma$ by $\Sigma''\finer\Sigma$.
      For all $t\in P$, we have $s\Sim_{(\Sigma,\porder)}^\star t$ and $t\Sim_{(\Sigma,\porder)}^\star s$,
      by $P$ stable on $(\Sigma,\porder)$. By definition we have $s\leq_{[s]_\Sigma}t$
      and $t\leq_{[s]_\Sigma}s$, and thus $t\in P'$.
      Therefore, $P\subseteq P'$. This proves $\Sigma''\finer\Sigma'$.

\item Let $P,Q\in\Sigma''$ and $P\porder''Q$. Since $\Sigma''\finer\Sigma'$, there exist $P', Q'\in\Sigma'$
      such that $P\subseteq P'$ and $Q\subseteq Q'$. We need to show that $P'\porder'Q'$.
      Taking $s_1\in P'$ and $s_2\in Q'$, we show that $s_1\Sim_{(\Sigma,\porder)}^\star s_2$.
      Let $t_1\in P$ and $t_2\in Q$, we have $t_1\Sim_{(\Sigma,\porder)}^\star t_2$. Also within
      $[s_1]_\Sigma$ we have $s_1\porder_{[s_1]_\Sigma}t_1$,
      and within $[s_2]_\Sigma$ we have $t_2\porder_{[s_2]_\Sigma}s_2$. As both $\porder_{[s_1]_\Sigma}$
      and $\porder_{[s_2]_\Sigma}$ are contained in $\Sim_{(\Sigma,\porder)}^\star$,
      We apply transitivity to get $s_1\Sim_{(\Sigma,\porder)}^\star s_2$. Therefore,
      $P'\porder'Q'$. This shows that $(P,Q)\in\ \porder'(\Sigma'')$, and thus $\porder''\subseteq\porder'(\Sigma'')$.

\qed
\end{itemize}
\end{proof}

The following lemma is used in the proof of Lemma~\ref{lem:rho-monotone}.

\begin{lemma}\label{lem:partition:order}
If $(\Sigma_1,\porder_1)\leq(\Sigma_2,\porder_2)$ and there are distributions $\Delta,\Delta'$
satisfying $\Delta_{\Sigma_1}\,\lift{\porder_1}\,\Delta'_{\Sigma_1}$, then
$\Delta_{\Sigma_2}\,\lift{\porder_2}\,\Delta'_{\Sigma_2}$.
\end{lemma}
\begin{proof} (sketch)
By reusing the same weight function for $\porder_1$ on the partition $\Sigma_1$ for
$\porder_2$ on the coarser partition $\Sigma_2$. \qed
\end{proof}

\begin{lemma}\label{lem:rho-monotone}
$\rho$ is monotonic on $(Part^0(\G), \leq)$.
\end{lemma}
\begin{proof}
Let $(\Sigma_1,\porder_1)\leq(\Sigma_2,\porder_2)$,
$(\Sigma'_1,\porder'_1)=\rho((\Sigma_1,\porder_1))$, $(\Sigma'_2,\porder'_2)=\rho((\Sigma_2,\porder_2))$.
We show that $(\Sigma'_1,\porder'_1)\,\leq\,(\Sigma'_2,\porder'_2)$.

We first prove that $(\Sigma'_1,\porder'_1)$ is stable on $(\Sigma_2,\porder_2)$.
Let $P,Q\in\Sigma_1'$ such that $P\porder_1' Q$, then for all $s\in P$, $t\in Q$ and $\pi\in\Pi_\pone$, there exists
$\pi'\in\Pi_\pone$ such that $\lift\step(s,\pi)\,\lift{\porder_1}\,\lift\step(t,\pi')$. Then by Lemma~\ref{lem:partition:order},
we also have $\lift\step(s,\pi)\,\lift{\porder_2}\,\lift\step(t,\pi')$.
By definition of $\rho$, we have that the partition pair $(\Sigma'_2,\porder'_2)$ is the unique largest
partition pair that is stable on $(\Sigma_2,\porder_2)$.
As $(\Sigma'_1,\porder'_1)$ is stable on $(\Sigma_2,\porder_2)$,
it must be the case that $(\Sigma'_1,\porder'_1)\,\leq\,(\Sigma'_2,\porder'_2)$. \qed
\end{proof}

Lemma~\ref{lem:stable-pa-sim} ensures that for all $(\Sigma,\porder)\in Part^0(\G)$,
$\Sim_{(\Sigma,\porder)}$ is a PA-\pone-simulation if $\rho((\Sigma,\porder))=(\Sigma,\porder)$,
i.e., $(\Sigma,\porder)$ is a fixpoint of $\rho$. However, we still need to find
the largest PA-\pone-simulation.
The following result indicates that if $S$ is finite, the coarsest stable partition pair
achieved by repetitively applying $\rho$
on $(\Sigma_0,\Id)$
indeed yields the largest PA-\pone-simulation.\footnote{
The following proof resembles the classical paradigm of finding the least fixpoint in an $\omega$-chain
of a complete partial order by treating $(\Sigma_0,\Id)$ as $\bot$. 
However, here we also need that fixpoint to represent the largest PA-\pone-simulation.}
Define $\rho^0(X)=X$ and $\rho^{n+1}(X)=\rho(\rho^n(X))$ for partition
pairs~$X$.

\begin{theorem}\label{thm:gcpp-max} 
Let $(\Sigma,\porder)=\bigcap_{i\in\Nat}\rho^i((\Sigma_0,\Id))$, then $\Sim_{(\Sigma,\porder)}$ is the largest PA-\pone-simulation on $\G$.
\end{theorem}
\begin{proof}(sketch)
Let $\Sim^+$ be the largest PA-\pone-simulation on $\G$. Define a set
$\Sigma^+=\set{\set{t\in S\mid s\Sim^+t\wedge t\Sim^+s}\mid s\in S}$. Since $\Sim^+$
is the largest PA-\pone-simulation, it can be shown that $\Sim^+$ is reflexive, symmetric
and transitive within each block $P\in\Sigma^+$. Moreover, we define a relation $\porder^+$
by $P\porder^+Q$ if there exists $s\in P$ and $t\in Q$ such that $s\Sim^+t$, and it can be
shown that $\porder^+$ is a partial order on $\Sigma^+$.
Then $(\Sigma^+,\porder^+)$ forms a partition pair on $\G$, and furthermore, it is stable,
and we also have $(\Sigma^+,\porder^+)\leq(\Sigma_0,\Id)$.

We apply $\rho$ on both sides. By Lemma~\ref{lem:rho-monotone} (monotonicity), and $(\Sigma^+,\porder^+)$
being stable, we have $(\Sigma^+,\porder^+)=\rho^i((\Sigma^+,\porder^+))\leq\rho^i((\Sigma_0,\Id))$ for all $i\in\Nat$.
As $Part(\G)$ is finite, there exists $j\in\Nat$, such that $\rho^j((\Sigma_0,\Id))=\rho^{j+1}((\Sigma_0,\Id))$.
Therefore, $\rho^j((\Sigma_0,\Id))$ is a stable partition pair, and $\Sim_{\rho^j((\Sigma_0,\Id))}$
is a PA-\pone-simulation by Lemma~\ref{lem:stable-pa-sim}. Straightforwardly we have
$\Sim^+\subseteq\Sim_{\rho^j((\Sigma_0,\Id))}$. Since $\Sim^+$ is the largest PA-\pone-simulation by assumption,
we have $\Sim^+=\Sim_{\rho^j((\Sigma_0,\Id))}$, and the result directly follows. \qed
\end{proof}

\section{A Decision Procedure for PA-\pone-Simulation}
\label{sec:pa-sim-alg}
Efficient algorithms for simulation in the non-probabilistic setting sometimes apply predecessor based methods~\cite{HHK95,GPP03}
for splitting blocks and refining partitions. This method can no longer be applied for simulations in the probabilistic setting,
as the transition functions now map a state to a state distribution rather than a single state, and simulation relation needs to
be \emph{lifted} to handle distributions. The algorithms in~\cite{ZHEJ08,Zhang08} follow the approaches in~\cite{BEM00} by reducing
the problem of deciding a weight function on lifted relations to checking the value of a maximal flow problem.
This method, however, does \emph{not} apply to combined transitions, where a more general solution is required.
Algorithms for deciding probabilistic bisimulations~\cite{CS02}
reduce the problem on checking weight functions with combined choices to solutions in linear programming (LP),
which are known to be decidable in polynomial time~\cite{Kar84}.\footnote{The maximal flow problem
is a special instance of an LP problem, which can be solved more efficiently~\cite{AMO93}.}

In our approach, simulation relations are characterised by partition pairs in the solutions to the GCPP.
Starting from the initial partition pair $(\Sigma_0,\Id)$, we gradually 
refine the partition by checking whether each pair of states in the same block can simulate
each other with respect to player \pone's choice on a chosen pivot block.
When deciding whether $s$ is able to simulate $t$, 
we need to examine potentially infinitely many mixed actions in $\Pi_\pone$.
This problem can be moderated by the following observations.
First we show that for $s$ to be simulated by $t$, it is only required to check
all deterministic choices of player \pone\ on $s$.

\begin{lemma}\label{lem:simplify1}
Let $(\Sigma,\porder)$ be a partition pair, then $t$ simulates $s$
on $(\Sigma,\porder)$ if for all $a\in\A_\pone$, there exists $\pi\in\Pi_\pone$ such that
$\lift\step(s,\pd{a})\,\lift\porder_{Sm}\,\lift\step(t,\pi)$.
\end{lemma}
\begin{proof}(sketch) By definition, $t$ simulates $s$ on $(\Sigma,\porder)$ if for all $\pi_1\in\Pi_\pone$
there exists $\pi_2\in\Pi_\pone$ such that $\lift\step(s,\pi_1)\lift\porder_{Sm}\lift\step(t,\pi_2)$.
Since $\pi_1(s)\in\dist(\A_\pone)$, for each $a_1\in\Supp{\pi_1(s)}$, we have some $\pi_3\in\Pi_\pone$ such that
$\lift\step(s,a_1)\lift\porder_{Sm}\lift\step(t,\pi_3)$, and we get $\pi_2$ by combining all such mixed actions $\pi_3$,
by applying Lemma~\ref{lem:lift-comb} and Lemma~\ref{lem:split:mix-action}.
\end{proof}

The next lemma says when checking a Smyth order $\lift\step(s,\pi)\,\lift\porder_{Sm}\,\lift\step(t,\pi')$,
it suffices to focus on player \ptwo's deterministic choices in $\lift\step(t,\pi')$
as all probabilistic choices of player \ptwo\ can be represented as interpolations from deterministic choices.

\begin{lemma}\label{lem:simplify2}
$\lift\step(s,\pi)\,\lift\porder_{Sm}\,\lift\step(t,\pi')$ if
for all $a\in\A_\ptwo$, there exists $\pi''\in\Pi_\ptwo$ such that
$\lift\step(s,\word{\pi,\pi''})\,\lift\porder\,\lift\step(t,\word{\pi',\pd{a}})$.
\end{lemma}
\begin{proof}(sketch) Similar to the proof of the above lemma by combining all the mixed actions in $\Pi_\ptwo$.
\end{proof}
Combining the above two lemmas, we have the following.
\begin{lemma}\label{lem:simplify3}
Let $(\Sigma,\porder)$ be a partition pair, then $t$ simulates $s$ with respect to player-\pone's choice
on $(\Sigma,\porder)$ if for all $a_1\in\A_\pone$, there exists $\pi_1\in\Pi_\pone$
such that for all $a_2\in\A_\ptwo$, there exists $\pi_2\in\Pi_\ptwo$ such that
$\lift\step(s,\word{\pd{a_1},\pi_2})\mathrel{\lift\porder}\lift\step(t,\word{\pi_1,\pd{a_2}})$.
\end{lemma}


The following lemma states how to check if the action $a$ can be followed by a mixed action from $\Pi_\pone$.
Given a finite set $S$ and $\porder$ a partial order on $S$, we define $\upclose{s}=\set{t\in S\mid s\porder t}$, called the \emph{up-closure} of $s$.
Finding a weight function for two distributions on a partition pair can be encoded in LP,  with linearity of the constraints guaranteed by Lemma~\ref{lem:simplify3}.

\begin{lemma}\label{lem:sim-lp}
Given a partition pair $(\Sigma,\porder)$, two states $s,t\in S$ and $a\in\A_\pone$,
there exists $\pi\in\Pi_\pone$ such that $\lift\step(s,\pd{a})\,\lift\porder_{Sm}\,\lift\step(t,\pi)$,
iff the following LP has a solution:\\
Let $\A_\pone=\set{a_1,a_2,\dots, a_\ell}$, $\A_\ptwo=\set{b_1,b_2,\dots, b_m}$ and $\Sigma=\set{B_1,B_2,\dots, B_n}$
\begin{equation}\label{eq:sum-alpha}\sum_{i=1}^\ell\alpha_i=1\end{equation}
\begin{equation}\label{eq:alpha}\forall i=1,2,\ldots, \ell: 0\leq\alpha_i\leq 1\end{equation}
\begin{equation}\label{eq:sum-beta}\forall j=1,2,\ldots, m:\sum_{k=1}^m\beta_{j,k}=1\end{equation}
\begin{equation}\label{eq:beta}\forall j,k=1,2,\ldots, m: 0\leq\beta_{j,k}\leq 1\end{equation}

\bigskip

\noindent Moreover, $\forall x,y=1,2,\ldots, n: j=1,2,\ldots, m:$
\begin{equation}\label{eq:weight-sum} 0\leq w_{x,y,j}\leq 1\end{equation}
\begin{equation}\label{eq:weight-left}\forall B_z\in\Sigma:\sum_{k=1}^m\beta_{j,k}\cdot\step(s,\word{a,b_k})(B_z) = \sum_{z'=1}^{n}w_{z,z',j} = \sum_{B_{z'}\in\upclose{B_z}}w_{z,z',j}\end{equation}
\begin{equation}\label{eq:weight-right}\forall B_z\in\Sigma:\sum_{i=1}^\ell\alpha_i\cdot\step(t,\word{a_{i},b_j})(B_z) = \sum_{z'=1}^{n}w_{z',z,j}\end{equation}
\end{lemma}
Informally, $\alpha_1,\alpha_2,\dots,\alpha_\ell$ are used to `guess' a mixed action from player \pone,
as constrained in Eq.~\ref{eq:sum-alpha} and Eq.~\ref{eq:alpha}.
To establish the Smyth order $\lift\porder_{Sm}$, 
for every player \ptwo\ action $b_j$ with $j=1,2,\dots, m$, we `guess' a mixed action from $\A_\ptwo$ represented
by $\beta_{j,1},\beta_{j,2}\ldots,\beta_{j,m}$, as constrained in Eq.~\ref{eq:sum-beta} and Eq.~\ref{eq:beta}.
Then for every player \ptwo\ action $b_j$, we use $w_{x,y,j}$ to represent the weight function that is required to establish
the lifted relation $\lift{\porder}$ for distributions $\sum_{k=1}^m\beta_{j,k}\cdot\delta(s,a,b_k)$
and $\sum_{i=1}^\ell\alpha_i\cdot\delta(t,a_{i},b_j)$, by Eq.~\ref{eq:weight-sum}, Eq.~\ref{eq:weight-left} and Eq.~\ref{eq:weight-right}.
In particular, the additional condition in Eq.~\ref{eq:weight-left} is to ensure that every non-zero $w_{z,z',j}$ must imply $B_z\porder B_{z'}$ as required by the weight function.
\begin{proof} (of Lemma~\ref{lem:sim-lp})
\begin{itemize}
\item[($\Leftarrow$)] Suppose the above LP has a solution, by Lemma~\ref{lem:simplify3},
we show that there exists a player $\pone$ mixed action $\pi\in\Pi_\pone$, such that for all $b_j\in\A_\ptwo$,
there exists a player $\ptwo$ mixed action $\sigma_j\in\Pi_\ptwo$ such that $\lift\step(s,\word{\pd{a},\sigma_j})\mathrel{\lift\porder}\lift\step(t,\word{\pi,\pd{b_j}})$.

From the solution of LP, a player $\pone$ mixed action $\pi$ can be defined by $\pi(a_i)=\alpha_i$ for all $a_i\in\A_\pone$,
satisfying $\sum_{i=1}^\ell\alpha_i=1$, by Eq.~\ref{eq:sum-alpha} and Eq.~\ref{eq:alpha}.
For each player $\ptwo$ action $b_j\in\A_\ptwo$, the mixed action $\sigma_i$ can be defined as $\sigma_i(b_k)=\beta_{j,k}$,
satisfying $\sum_{k=1}^m\beta_{j,k}=1$, by Eq.~\ref{eq:sum-beta} and Eq.~\ref{eq:beta}.
Next we show for all $1\leq j\leq m$, $\lift\step(s,\word{\pd{a},\sigma_j})\mathrel{\lift\porder}\lift\step(t,\word{\pi,\pd{b_j}})$,
which is equivalent to $\sum_{k=1}^m\sigma_j(b_k)\cdot\step(s,\word{a,b_k})\mathrel{\lift\porder}\sum_{i=1}^\ell\pi(a_i)\cdot\step(t,\word{a_i, b_j})$
by Lemma~\ref{lem:split:mix-action} and Lemma~\ref{lem:split:mix-action-2}.
Given the partition $\Sigma=\set{B_1,B_2,\dots, B_n}$,
a weight function $w:\Sigma\times\Sigma\rightarrow[0,1]$ can be defined by $w(B_x, B_y)=w_{x,y,j}$ for all $1\leq x,y\leq n$.
The conditions on weighted sums are given in Eq.~\ref{eq:weight-left} and Eq.~\ref{eq:weight-right}.
We show that $w(B_x,B_y)=w_{x,y,j}>0$ implies $B_x\porder B_y$.
Suppose $w(B_x, B_y)>0$ and $B_x\not\porder B_y$,
then $B_y\not\in\upclose{B_x}$, which would imply $\sum_{x'=1}^{n}w_{x,x',j} > \sum_{B_{x'}\in\upclose{B_x}}w_{x,x',j}$,
contradicting Eq.~\ref{eq:weight-left}.

\item[($\Rightarrow$)]  Suppose there exists a player $\pone$ mixed action $\pi\in\Pi_\pone$,
such that for all $b_j\in\A_\ptwo$, there exists a player $\ptwo$ mixed action $\sigma_j\in\Pi_\ptwo$
such that $\lift\step(s,\word{\pd{a},\sigma_j})\mathrel{\lift\porder}\lift\step(t,\word{\pi,\pd{b_j}})$.
By Lemma~\ref{lem:split:mix-action} and Lemma~\ref{lem:split:mix-action-2},
equivalently, we have $\sum_{k=1}^m\sigma_j(b_k)\cdot\step(s,\word{a,b_k})\mathrel{\lift\porder}\sum_{i=1}^\ell\pi(a_i)\cdot\step(t,\word{a_i, b_j})$.
We show the above LP constraints have a solution.

First we let $\alpha_i=\pi(a_i)$ for each $a_i\in\A_\pone$, which satisfies Eq.~\ref{eq:sum-alpha} and Eq.~\ref{eq:alpha}.
Similarly, for each $1\leq j\leq m$, let $\beta_{j,k}=\sigma_j(b_k)$ for all $1\leq k\leq m$, which satisfies Eq.~\ref{eq:sum-beta} and Eq.~\ref{eq:beta}.
For each $1\leq j\leq m$, 
given the partition $\Sigma=\set{B_1,B_2,\dots, B_n}$, the existing weight function $w:\Sigma\times\Sigma\rightarrow[0,1]$ satisfies
\begin{itemize}
 \item[(a)] for all $B_x\in\Sigma$, $\sum_{y=1}^{n}w(B_x,B_y)=\sum_{k=1}^m\sigma_j(b_k)\cdot\step(s,\word{a,b_k})(B_x)$,
 \item[(b)] for all $B_y\in\Sigma$, $\sum_{x=1}^{n}w(B_x,B_y)=\sum_{i=1}^\ell\pi(a_i)\cdot\step(t,\word{a_i, b_j})(B_y)$,
 \item[(c)] for all $B_x, B_y\in\Sigma$, $w(B_x,B_y)>0$ implies $B_x\porder B_y$.
\end{itemize}
For each $1\leq j\leq m$ we let $w_{x,y,j} = w(B_x, B_y)$.
It is then clear that after replacing each $\sigma_j(b_k)$ by $\beta_{j,k}$ in (a), we get the left equality of Eq.~\ref{eq:weight-left}.
Similarly, after replacing each $\pi(a_i)$ by $\alpha_i$ in (b) we get Eq.~\ref{eq:weight-right}.
Next we show that $\sum_{y=1}^{n}w_{x,y,j} = \sum_{B_{y}\in\upclose{B_x}}w_{x,y,j}$.
First we have $\sum_{B_{y}\in\upclose{B_x}}w_{x,y,j}\leq \sum_{y=1}^{n}w_{x,y,j}$ by $\upclose{B_x}\subseteq\Sigma$.
If $\sum_{B_{y}\in\upclose{B_x}}w_{x,y,j}<\sum_{y=1}^{n}w_{x,y,j}$,
there would be some $B_y\not\in\upclose{B_x}$ and $w_{x,y,j}>0$, which implies $w(B_x,B_y)>0$ and $B_x\not\porder B_y$, contradicting (c).
\end{itemize}
\end{proof}

We define a predicate $\CanFollow$ such that $\CanFollow((\Sigma,\porder),s,t,a)$ decides
whether there exists a mixed action of player \pone\ from $t$ which simulates action $a\in\A_\pone$
from $s$ on the partition pair $(\Sigma,\porder)$. $\CanFollow$ establishes 
an LP problem from its parameters (see Lemma~\ref{lem:sim-lp}).
We further define a predicate $\CanSim$ which decides whether a state simulates another
with respect to player~\pone's choice on $(\Sigma,\porder)$ for all actions in $\A_\pone$,
i.e., ${\CanSim((\Sigma,\porder),s,t)}$ returns \emph{true} if $\CanFollow((\Sigma,\porder),s,t,a)$
returns \emph{true} for~all~${a\in\A_\pone}$.

\begin{algorithm}[!ht]
\caption{\label{alg:split} Refining a block to make it stable on a partition pair}
\begin{algorithmic}
	\STATE INPUT: a partition pair $(\Sigma,\porder)$, a block $B\in\Sigma$
	\STATE OUTPUT: a partition pair $(\Sigma_B,\porder_B)$ on $B$
	\STATE {\bf function} {\sf Split} ($(\Sigma,\porder)$, $B$)
	\STATE
	\STATE {\bf begin}
	\STATE \hspace{10pt} $\Sigma_B:=\set{\set{s}\mid s\in B}$; $\porder_B:=\set{(s,s)\mid s\in B}$; $\Sigma':=\emptyset$; $\porder':=\emptyset$
	\STATE \hspace{10pt} \while\ $\Sigma_B\neq\Sigma'\vee\porder_B\neq\porder'$ \doo
	\STATE \hspace{20pt} $\Sigma':=\Sigma_B$; $\porder':=\porder_B$
	\STATE \hspace{20pt} \forEach\ \textit{distinct} $B_1,B_2\in\Sigma_B$ \doo
         \STATE \hspace{30pt} \textit{pick any} $s_1\in B_1$ \textit{and} $s_2\in B_2$
	\STATE \hspace{30pt} {\bf if} ($\CanSim((\Sigma,\porder),s_1,s_2)\wedge\CanSim((\Sigma,\porder),s_2,s_1)$) {\bf then}
	\STATE \hspace{40pt} $\Sigma_B:=\Sigma_B\setminus\set{B_1,B_2}\cup\set{B_1\cup B_2}$
	\STATE \hspace{40pt} $\porder_B:=\porder_B\cup\ \set{(X,B_1\cup B_2)\mid X\in\Sigma: (X,B_1)\in\porder_B\vee\ (X,B_2)\in\porder_B}$
	\STATE \hspace{60pt} $\cup\set{(B_1\cup B_2,X)\mid X\in\Sigma: (B_1,X)\in\porder_B\vee\ (B_2,X)\in\porder_B}$
	\STATE \hspace{60pt} $\setminus\set{(B_i,X),(X,B_i)\mid X\in\Sigma: (B_i,X),(X,B_i)\in\porder_B\wedge\ i\in\set{1,2}}$
	\STATE \hspace{30pt} {\bf else if} ($\CanSim((\Sigma,\porder),s_1,s_2)$) {\bf then}
	\STATE \hspace{40pt} $\porder_B:=\porder_B\cup\,\set{(B_2,B_1)}$
	\STATE \hspace{30pt} {\bf else if} ($\CanSim((\Sigma,\porder),s_2,s_1)$) {\bf then}
	\STATE \hspace{40pt} $\porder_B:=\porder_B\cup\,\set{(B_1,B_2)}$
	\STATE \hspace{20pt} {\bf endfor}
	\STATE \hspace{10pt} {\bf endwhile}
	\STATE \hspace{10pt} {\bf return} $(\Sigma_B,\porder_B)$
	\STATE {\bf end}
\end{algorithmic}
\end{algorithm}

Algorithm~\ref{alg:split} defines a function {\sf Split} which refines a block $B\in\Sigma$ into a partition
pair corresponding to the maximal simulation that is stable on $(\Sigma,\porder)$. 
It starts with the finest partition and the identity relation (as the final relation is reflexive).
For each pair of blocks in the partition, we check if they can simulate each other by picking up
a state from each block.
(The choice of a state is arbitrary, because all states within the same block are simulation equivalent on $(\Sigma,\porder)$.)
If the two state are simulation equivalent on $(\Sigma,\porder)$ then we merge the two
blocks as well as all incoming and outgoing relation in the current partial order. If only one simulates
the other we add an appropriate pair into the current ordering. 
This process continues until~the~partition~pair~stabilises, when no more merging of partitions can happen
or any more pair can be added to $\porder_B$, which means the resulting partition pair $(\Sigma_B,\porder_B)$ is maximal.

Algorithm~\ref{alg:gcpp} is based on the functionality of \textsf{Split} in Algorithm~\ref{alg:split}.
Starting from the partition $(\Sigma_0,\Id)$, which is identified as
$(\set{\set{t\mid\L(t)=\L(s)}\mid s\in S}, \set{(B,B)\mid B\in\Sigma_0})$,
the algorithm computes a sequence of partition pairs $(\Sigma_1,\porder_1), (\Sigma_2,\porder_2)\ldots$
until it stabilises, which is detected by checking the condition $\Sigma\neq\Sigma'\,\vee\,\porder\neq\porder'$.
At each iteration we have $(\Sigma_{i+1},\porder_{i+1})\leq (\Sigma_i,\porder_i)$. Moreover,
$(\Sigma_{i+1},\porder_{i+1})$ is the maximal partition pair that is stable on $(\Sigma_i,\porder_i)$,
because by Algorithm~\ref{alg:split}, the splitting of each block $B$ in $\Sigma_i$ creates 
a maximal partition pair $(\Sigma_B,\porder_B)$ of $B$
that is stable on $(\Sigma_i,\porder_i)$, and the new partition pair $(\Sigma_{i+1},\porder_{i+1})$ is
formed by merging all such maximal pairs as well as by taking into account the previous relation represented
by $(\Sigma_i,\porder_i)$. Intuitively, we have $(\Sigma_{i+1},\porder_{i+1})=\rho((\Sigma_i,\porder_i))$,
where $\rho$ is the operator as per Definition~\ref{def:rho}. %
The correctness of the algorithm is then justified by Theorem~\ref{thm:gcpp-max}, which states that it converges to the
coarsest partition pair that is contained in $(\Sigma_0,\Id)$ and
returns a representation of
the largest PA-\pone-simulation.

\begin{algorithm}[!ht]
\caption{\label{alg:gcpp} Computing the Generalised Coarsest Partition Pair}
\begin{algorithmic}
	\STATE INPUT: a probabilistic game structure $\G=\word{S, s_0, \L, \A,\step}$
	\STATE OUTPUT: a partition pair $(\Sigma,\porder)$ on $S$
	\STATE {\bf function} {\sf GCPP} ($\G$)
	\STATE
	\STATE {\bf begin}
	\STATE \hspace{10pt} $\Sigma:=\set{\set{t\mid\L(t)=\L(s)}\mid s\in S}$; $\porder\,:=\set{(B,B)\mid B\in\Sigma}$
    \STATE \hspace{10pt} $\Sigma':=\emptyset$; $\porder':=\emptyset$
	\STATE \hspace{10pt} \while\ $\Sigma\neq\Sigma'\vee\porder\neq\porder'$ \doo
	\STATE \hspace{20pt} $\Sigma':=\Sigma$; $\porder':=\porder$
	\STATE \hspace{20pt} \forEach\ $B\in\Sigma$ \doo
    \STATE \hspace{30pt} $(\Sigma_B,\porder_B):=\textsf{Split}((\Sigma',\porder'), B)$
    \STATE \hspace{30pt} $\Sigma:=\Sigma\setminus\set{B}\cup\Sigma_B$
    \STATE \hspace{30pt} $\porder\ :=\ \porder\cup\porder_B$
    \STATE \hspace{50pt} $\cup\,\set{(B',X)\mid X\in\Sigma: B'\in\Sigma_B: (B,X)\in\,\porder}$
    \STATE \hspace{50pt} $\cup\,\set{(X,B')\mid X\in\Sigma: B'\in\Sigma_B: (X,B)\in\,\porder}$
    \STATE \hspace{50pt} $\setminus\set{(B,X),(X,B)\mid X\in\Sigma: (X,B),(B,X)\in\,\porder}$
	\STATE \hspace{20pt} {\bf endfor}
	\STATE \hspace{10pt} {\bf endwhile}
	\STATE \hspace{10pt} {\bf return} $(\Sigma,\porder)$
	\STATE {\bf end}
\end{algorithmic}
\end{algorithm}

\paragraph{Space complexity.}
For a PGS $\word{S,s_0,\L,\A,\step}$, it requires $\bigO(|S|)$ to store
the state space and $\bigO(|S|^2\cdot|\A|)$ for the transition relation, since for each $s\in S$
and $\word{a_1,a_2}\in\A$ it requires an array of size $\bigO(|S|)$ to store a distribution.
Recording a partition pair takes $\bigO(|S|^2 +|S|^2)$ as
the first part is needed to record for each state which equivalence class in the
partition it belongs, 
and the second part is needed for the partial order relation $\porder$.
The computation from $(\Sigma_i,\porder_i)$ to $(\Sigma_{i+1},\porder_{i+1})$ can be done in-place
which only requires additional constant space to track if the partition pair has been modified during each iteration.
Another extra space-consuming part is for solving LP constrains, which we assume has space
usage $\bigO(\gamma(N))$ where $N=1+|\A_\pone|+|\A_\ptwo|+|\A_\ptwo|^2+  |S|^2\cdot|\A_\ptwo|  + 3\cdot |S|\cdot|\A_\ptwo|$
is the number of linear constraints at most, and $\gamma(N)$ some polynomial. The space complexity roughly
sums up to $\bigO(|S|^2\cdot|\A|+\gamma(|\A|^2+|S|^2\cdot|\A|))$.
The first part $\bigO(|S|^2\cdot|\A|)$ for the PGS itself can be considered optimal,
while the second part depends on the efficiency of the LP algorithm being used.

\paragraph{Time complexity.}
The number of variables in the LP problem in Lemma~\ref{lem:sim-lp} is $|\A_\pone| + |\A_\ptwo|^2 + |S|^2\cdot|\A_\ptwo|$,
and the number of constraints is bounded by $1+|\A_\pone|+|\A_\ptwo|+|\A_\ptwo|^2+ |S|^2\cdot|\A_\ptwo|  + 3\cdot |S|\cdot|\A_\ptwo|$. 
The predicate $\CanSim$ costs $|\A_\pone|$ times LP solving.
Each {\sf Split} invokes at most $|B|^2$ testing of $\CanSim$ where $B$ is a block in $\Sigma$.
Each iteration of {\sf GCPP} splits all current blocks,
and the total number of comparisons within each iteration of {\sf GCPP} is be bounded
by $|S|^2$.
(However it seems heuristics on the existing partition can achieve a speed close
to linear in practice by caching previous \CanSim\ checks~\cite{ZHEJ08}.)
The number of iterations is bounded by $|S|$. This gives us time complexity
which is in the worst case to solve
$\bigO(|S|^3\cdot|\A_\pone|)$ many such LP problems, each of which has
$\bigO(|\A|^2+ |S|^2\cdot|\A|)$ constraints.

\paragraph{Remark.}
By removing the interaction between players (i.e., the alternating
part), our algorithm downgrades to a partition-based algorithm computing the largest \emph{strong}
probabilistic simulation relation in probabilistic automata, where \emph{combined transitions} are needed.
This simplified setting is equivalent to removing choices from player $\ptwo$ from PGS. (Informally, we let $|\A_\ptwo| = 1$.)
Now the time complexity is to solve $\bigO(|S|^3\cdot|\A|)$ many such LP problems,
each of which has $\bigO(|\A|+ |S|^2)$ constraints.
The algorithm of~\cite{ZHEJ08} for computing strong probabilistic simulation
has time complexity of solving $\bigO(|S|^2\cdot m)$ LP problems, where $m$ is the size of the
transition relation comparable to $\bigO(|S|^2\cdot|\A|)$.
They have $\bigO(|S|^2)$ constraints for each LP instance.
Note that the space-efficient algorithm~\cite{Zhang08} for probabilistic simulation ({\em without} combined transitions)
has the same space complexity but better time complexity than ours,
which is due to the reduction to the maximal~flow~problem.

\section{Conclusion}
\label{sec:conclusion}

We have presented a partition-based algorithm to compute the largest probabilistic alternating
simulation relation in finite probabilistic game structures. To the best of our knowledge, our work presents
the first polynomial-time algorithm for computing a relation in probabilistic systems considering 
(concurrently) mixed choices from players. 
As aforementioned, PA-simulation is known as stronger
than the simulation relation characterising quantitative $\mu$-calculus~\cite{dAMRS08},
though it is still a conservative approximation which has a reasonable complexity to be
useful in verification of game-based properties. 

\subsection*{Acknowledgement}
Special thank goes to Timothy Bourke who carefully read through the draft, with his effort greatly 
improving the presentation of the paper.
The authors also thank Wan Fokkink, Rob van Glabbeek and Lijun Zhang for their helpful comments on
the technical part. 


\begin{thebibliography}{10}

\bibitem{AMO93}
R.~K.~Ahuja, T.~L.~Magnanti, and J.~B.~Orlin.
\newblock {\em Network Flows: Theory, Algorithms, and Applications}.
\newblock Prentice Hall, 1993.

\bibitem{AHK02}
R.~Alur, T.~A. Henzinger, and O.~Kupferman.
\newblock Alternating-time temporal logic.
\newblock {\em Journal of ACM}, 49(5):672--713, 2002.

\bibitem{AHKV98}
R.~Alur, T.~A. Henzinger, O.~Kupferman, and M.~Y. Vardi.
\newblock Alternating refinement relations.
\newblock In {\em Proc.\ CONCUR}, LNCS 1466, pages 163--178. Springer, 1998.

\bibitem{BEM00}
C.~Baier, B.~Engelen, and M.~E. Majster-Cederbaum.
\newblock Deciding bisimilarity and similarity for probabilistic processes.
\newblock {\em Journal of Computer and System Sciences}, 60(1):187--231, 2000.

\bibitem{BG03}
D.~Bustan and O.~Grumberg.
\newblock Simulation based minimization.
\newblock {\em ACM Transactions on Computational Logic}, 4(2):181--206, 2003.

\bibitem{CS02}
S.~Cattani and R.~Segala.
\newblock Decision algorithms for probabilistic bisimulations.
\newblock In {\em Proc.\ CONCUR}, LNCS 2421, pages 371--386. Springer, 2002.

\bibitem{CDH06}
K.~Chatterjee, L.~de~Alfaro, and T.~A. Henzinger.
\newblock The complexity of quantitative concurrent parity games.
\newblock In {\em Proc.\ SODA},  pages 678--687. ACM, 2006.

\bibitem{CdAMR10}
K.~Chatterjee, L.~{de Alfaro}, R.~Majumdar, and V.~Raman.
\newblock Algorithms for game metrics (full version).
\newblock {\em Logical Methods in Computer Science}, 6(3:13):1--27, 2010.

\bibitem{CE81}
E.~M. Clarke and E.~A. Emerson.
\newblock Synthesis of synchronization skeletons for branching-time temporal
  logic.
\newblock In {\em Proc.\ Workshop on Logics of Programs}, LNCS 131, pages 52--71. Springer, 1981.

\bibitem{CR11}
S.~Crafa and F.~Ranzato.
\newblock Probabilistic bisimulation and simulation algorithms by abstract interpretation.
\newblock In {\em Proc.\ 38th Colloquium on Automata, Languages and
  Programming}, LNCS 6756, pages 295--306. Springer, 2011.

\bibitem{Alf03}
L.~de~Alfaro.
\newblock Quantitative verification and control via the mu-calculus.
\newblock In {\em Proc.\ CONCUR}, LNCS 2761, pages 102--126. Springer, 2003.

\bibitem{AHK98}
L.~de~Alfaro, T.~A. Henzinger, and O.~Kupferman.
\newblock Concurrent reachability games.
\newblock In {\em Proc.\ FOCS}, pages 564--575. IEEE CS, 1998.

\bibitem{AM04}
L.~de~Alfaro and R.~Majumdar.
\newblock Quantitative solution of omega-regular games.
\newblock {\em Journal of Computer and System Sciences}, 68(2):374--397, 2004.

\bibitem{dAMRS08}
L.~{de Alfaro}, R.~Majumdar, V.~Raman, and M.~Stoelinga.
\newblock Game refinement relations and metrics.
\newblock {\em Logic Methods in Computer Science}, 4(3:7):1--28, 2008.

\bibitem{GPP03}
R.~Gentilini, C.~Piazza, and A.~Policriti.
\newblock From bisimulation to simulation: {C}oarsest partition problems.
\newblock {\em Journal of Automatic Reasoning}, 31(1):73--103, 2003.

\bibitem{GL94}
O.~Grumberg and D.~Long.
\newblock Model checking and modular verification.
\newblock {\em ACM Transactions on Programming Languages and Systems},
  16(3):843--871, 1994.

\bibitem{HHK95}
M.~R. Henzinger, T.~A. Henzinger, and P.~W. Kopke.
\newblock Computing simulations on finite and infinite graphs.
\newblock In {\em Proc.\ FOCS}, pages 453--462. IEEE CS, 1995.

\bibitem{JL91}
B.~Jonsson and K.~G. Larsen.
\newblock Specification and refinement of probabilistic processes.
\newblock In {\em Proc.\ LICS}, pages 266--277. IEEE CS, 1991.

\bibitem{Kar84}
N.~Karmakar.
\newblock A new polynomial-time algorithm for linear programming.
\newblock {\em Combinatorica}, 4(4):373--395, 1984.

\bibitem{Kozen83}
D.~C. Kozen.
\newblock Results on the propositional $\mu$-calculus.
\newblock {\em Theoretical Computer Science}, 27:333--354, 1983.

\bibitem{RT07}
F.~Ranzato and F.~Tapparo.
\newblock A new efficient simulation equivalence algorithm.
\newblock In {\em Proc.\ LICS}, pages 171--180. IEEE CS, 2007.

\bibitem{Sch86}
A.~Schrijver.
\newblock {\em Theory of Linear and Integer Programming}.
\newblock Wiley, 1986.

\bibitem{Seg95}
R.~Segala.
\newblock {\em Modeling and Verification of Randomized Distributed Real-Time
  Systems}.
\newblock PhD thesis, Massachusetts Institute of Technology, 1995.

\bibitem{SL95}
R.~Segala and N.~A. Lynch.
\newblock Probabilistic simulations for probabilistic processes.
\newblock {\em Nordic Journal of Computing}, 2(2):250--273, 1995.

\bibitem{Shapley53}
L.~S. Shapley.
\newblock Stochastic games.
\newblock {\em Proc. National Academy of Science}, 39:1095--1100, 1953.

\bibitem{Smy78}
M.~B.~Smyth.
\newblock Power domains.
\newblock {\em Journal of Computer and System Sciences}, 16(1):23--36, 1978.

\bibitem{TC01}
L.~Tan and R.~Cleaveland.
\newblock Simulation revisited.
\newblock In {\em Proc.\ TACAS}, LNCS 2031, pages 480--495. Springer, 2001.

\bibitem{GP08}
R.~J. van Glabbeek and B.~Ploeger.
\newblock Correcting a space-efficient simulation algorithm.
\newblock In {\em Proc.\ CAV}, LNCS 5123, pages 517--529. Springer, 2008.

\bibitem{NM47}
J.~von Neumann and O.~Morgenstern.
\newblock {\em Theory of Games and Economic Behavior}.
\newblock Princeton University Press, 1947.

\bibitem{ZP10}
C.~Zhang and J.~Pang.
\newblock On probabilistic alternating simulations.
\newblock In {\em Proc.\ IFIP TCS}, AICT 323, pages 71--85. IFIP, 2010.

\bibitem{Zhang08}
L.~Zhang.
\newblock A space-efficient probabilistic simulation algorithm.
\newblock In {\em Proc.\ CONCUR}, LNCS 5201, pages 248--263. Springer, 2008.

\bibitem{ZHEJ08}
L.~Zhang, H.~Hermanns, F.~Eisenbrand, and D.~N. Jansen.
\newblock Flow faster: {E}fficient decision algorithms for probabilistic
  simulations.
\newblock {\em Logical Methods in Computer Science}, 4(4:6):1--43, 2008.

\end{thebibliography}

\end{document}